\documentclass[aps,prx,reprint,superscriptaddress,nofootinbib,
               amsmath,amssymb,longbibliography]{revtex4-2}

\usepackage{amsmath,amssymb,amsthm,graphicx,bm}
\usepackage[colorlinks=true,linkcolor=blue,citecolor=blue,urlcolor=blue]{hyperref}

\theoremstyle{plain}
\newtheorem{theorem}{Theorem}

\newtheorem{lemma}[theorem]{Lemma}
\newtheorem{corollary}[theorem]{Corollary}
\theoremstyle{remark}

\newcommand{\Wu}{W^{u}}
\newcommand{\Wss}{W^{ss}}
\newcommand{\Cs}{C^{s}_{\rho}}

\newcommand{\Bu}{\mathcal{B}^{u}_{\nu}}

\newcommand{\Pin}{\mathcal{P}_{\nu}}

\newcommand{\Real}{\mathbb{R}}

\begin{document}

\title{An Arrow of Time in Stable Autonomous Generation}
\author{Zhixin Lu}
\affiliation{Allen Institute, Seattle, Washington 98109, USA}
\date{October 1, 2026}

\begin{abstract}
Does stable autonomous generation select a direction of time?
This question matters both for learning models that autonomously
generate prescribed dynamics and for designing controllers that
make an invariant set attracting without changing its dynamics.
Periodic motion and quasiperiodic torus rotations admit stable
generation in either direction. For a chaotic attractor, however,
reversing the equations reproduces its backward trajectories while
turning attraction into repulsion. Can additional latent variables
or autonomous feedback restore attraction without changing the
reversed motion? We prove an obstruction to such exact generation
by finite-dimensional autonomous systems with $C^2$ vector fields,
compact attracting sets, and continuous readouts. The result applies
to the classical Lorenz attractor and, more generally, to dynamics
with an ergodic invariant measure of positive metric entropy whose
intrinsic local stable fibers are totally disconnected on a set of
positive measure. The proof connects this geometry to an entropy
contradiction: a hypothetical attracting realization would force
the output entropy to vanish, contrary to the prescribed chaotic
dynamics. We use reservoir-computing examples to illustrate the
difficulty of stable autonomous generation for time-reversed chaos:
both time directions are closely tracked while the networks are
driven by data, but only the forward example retains Lorenz geometry
during autonomous generation. The theorem identifies a class of dynamics 
that cannot be learned for exact, stable autonomous generation 
by finite-dimensional smooth models satisfying the stated 
assumptions, regardless of the learning procedure or the number 
of latent variables.
\end{abstract}
\maketitle
\raggedbottom

\section{Introduction}
\label{sec:intro}

Machine-learning models can learn to generate complex temporal patterns,
from speech and music to chaotic dynamical signals. People likewise learn
to play melodies and coordinate intricate movements. Yet a pattern that
we readily produce in its familiar order can be difficult to perform
backward. Although experience, learning, and physical constraints may
explain this difficulty, it raises a question about the dynamics
themselves: can stable autonomous generation be feasible in only one
direction of time?

Earlier work investigates time irreversibility by comparing the
statistics of observed sequences with those of their
reversals~\cite{Pomeau1982,Diks1995}, including through directed
visibility graphs constructed from time series~\cite{Lacasa2012}.
These methods ask whether the two temporal orderings are statistically
distinguishable. We ask a different question: can dynamics that are generated 
stably in forward time also be generated stably in reverse by another 
dynamical system? Our theorem rules out such a possibility for a class 
of dynamics under the conditions stated below.

To formulate the question, let $\dot x=F(x)$ have flow $\phi_t$ and a
compact attractor $A$. Can another dynamical system stably generate the
reversed trajectories $\phi_{-t}(x)$ on $A$? Simply replacing $F$ by
$-F$ produces the desired trajectories but turns attraction into
repulsion. A different vector field might restore attraction while
preserving the reversed motion on $A$. For a limit cycle or a
quasi-periodic torus rotation, one can reverse the phase motion while
retaining contraction normal to the invariant manifold. Whether this
stabilization remains possible for a chaotic attractor is less clear.

Consider first adding a stabilizing correction in the original state space,
\[
\dot x=-F(x)+R(x),\qquad R|_A=0.
\]
Here $R$ is $C^1$ near $A$, and $R|_A=0$ preserves every reversed
trajectory on $A$. For a smooth limit cycle, such a correction can vanish
on the cycle while providing attraction toward it. The Lorenz attractor,
however, has a layered geometry. Its
strong-stable direction contracts forward in time and expands under
reversal~\cite{APPV2009}, suggesting where a correction might try to
restore contraction. The neighboring layers constrain this correction:
if distinct points $x_k\in A\setminus\{x\}$ approach $x$ with
\[
\frac{x_k-x}{\|x_k-x\|}\longrightarrow e,
\]
then $R(x_k)=R(x)=0$ and differentiability imply
\[
D_e R(x)
:= \left.\frac{d}{dh}R(x+he)\right|_{h=0}
= 0.
\]
Wherever a formerly contracting direction occurs as such an accumulation
direction, the correction term cannot supply additional linear attraction along
it at that point for the time-reversed dynamics. This local constraint
motivates the obstruction but does not, by itself, prove that the whole
set cannot be made attracting.

One may also try to stabilize the reversed motion by adding latent
variables. Consider an autonomous generator $\dot z=G(z)$ with flow $\Psi_t$,
for example on $\Real^D$. We require a compact attracting set
$\Lambda$ in the full generator state space, so attraction includes
all latent variables. A continuous readout $\pi:\Lambda\to A$ converts the generator
state into the target state. The condition $\pi(\Lambda)=A$
ensures that every target state is represented. We also require
every generator trajectory on $\Lambda$ to produce the prescribed
reversed target trajectory:
\[
\pi(\Psi_t(z))=\phi_{-t}(\pi(z))
\qquad(z\in\Lambda,\ t\in\mathbb R).
\]
This formulation is relevant to recurrent neural networks (RNNs),
whose internal state can have many more dimensions than the observed
dynamics. RNNs can predict chaotic trajectories and reproduce long-term
statistics~\cite{Jaeger2004,Pathak2018,LuHuntOtt2018}. Learning dynamics
for autonomous generation aims to acquire an internal dynamical realization
that continues to produce the prescribed behavior after external input is
removed~\cite{LuBassett2020}. Delay-embedding and reservoir-embedding
results show that, under suitable hypotheses, histories of observations
or states of a driven reservoir can uniquely identify the underlying
dynamical state~\cite{Takens1981,Sauer1991,Hart2020}. Such a representation
does not by itself guarantee attraction when the model runs autonomously.
Our theorem identifies a geometric obstruction to exact attracting
realization within the stated generator class. The reservoir-computing
examples illustrate the distinction between driven tracking and autonomous
generation for forward and reversed Lorenz signals.

The same existence question arises in autonomous feedback control. With full state
$z=(x,v)$, consider
\[
\dot x=-F(x)+R(x,v),\qquad \dot v=H(x,v).
\]
Here the readout is the coordinate projection $\pi(x,v)=x$. We ask whether
$R$ and $H$ can produce a compact attracting set $\Lambda$ with
$\pi(\Lambda)=A$ and $R(x,v)=0$ at every point of $\Lambda$. The latter
condition preserves the exact reversed motion, while attraction
concerns the joint state $z=(x,v)$. 

We prove an obstruction to such exact reversed stable generation by any
finite-dimensional system with $C^2$ autonomous vector field, compact
attracting set, and continuous readout. The obstruction applies if,
for some $\tau>0$, the time-$\tau$ map of the target flow has positive
metric entropy with
respect to an ergodic invariant probability measure $\mu$ and has totally
disconnected local stable fibers (defined in Eq.~\eqref{eq:fiber})
on a subset $E\subset A$ with positive measure $\mu(E)>0$.

The proof proceeds by contradiction. We construct an ergodic invariant
measure $\nu$ for the hypothetical generator system with $\pi_*\nu=\mu$.
For almost every generator state whose output lies in $E$, attraction and
continuity force the readout to be constant on each complete unstable leaf.
Target ergodicity extends the fiber condition to $\mu$-almost every target
point, giving whole-leaf collapse at $\nu$-almost every generator state.
Ledrappier-Young's Theorem~B~\cite{LY1985} then forces the output entropy
to vanish, contradicting the prescribed positive entropy.

The classical Lorenz attractor satisfies these conditions. Its forward
equations provide an attracting realization, whereas our theorem
excludes an exact and stable reversed generation in the specified 
smooth class of generator systems, even with additional latent variables. 
This establishes an \emph{arrow of time in stable autonomous generation}
for Lorenz dynamics and, more generally, for smooth attracting systems
satisfying the theorem's hypotheses. 

Section~\ref{sec:setup} states the obstruction and its flow corollary,
and Section~\ref{sec:proof} gives the complete proof.
Section~\ref{sec:lorenz} establishes the Lorenz application and presents
the reservoir-computing examples. Section~\ref{sec:discussion} discusses
the implications for learning and control.

\section{The obstruction}
\label{sec:setup}

We first state the obstruction for discrete-time dynamics. Let
$f:A\to A$ be a homeomorphism of a compact metric space with
metric $d_A$. The generator evolves on a finite-dimensional smooth
manifold $M$, equipped with a smooth Riemannian metric and its
associated distance $d_M$, and has a compact attracting set
$\Lambda\subset M$. No smooth-manifold structure is assumed for
$\Lambda$; derivatives and tangent directions are taken in the
ambient manifold $M$. We write $x\in A$ for a target state and
$z\in\Lambda$ for a generator state on the attracting set. The
invariant target measure is denoted by $\mu$; an invariant generator
measure $\nu$ will be constructed in the proof.

For a radius $\rho>0$, the \emph{intrinsic local forward-stable fiber}
of $x$ consists of points in $A$ whose forward trajectories remain
within distance $\rho$ of the trajectory of $x$ and converge toward it:
\begin{equation}
\Cs(x)=\left\{x'\in A:
\begin{aligned}
&d_A(f^n(x'),f^n(x))<\rho\;(n\in\mathbb Z_{\geq 0}),\\
&d_A(f^n(x'),f^n(x))\longrightarrow0
  \;\text{as }n\to\infty
\end{aligned}\right\}.
\label{eq:fiber}
\end{equation}
The convergence is synchronous: the two trajectories are compared
after the same number of iterations. The term \emph{intrinsic}
emphasizes that the fiber is defined entirely within $A$, using
only $f$ and $d_A$.

A set is \emph{totally disconnected} if every connected subset has
at most one point. For Lorenz dynamics, Sec.~\ref{sec:lorenz}
connects the intrinsic fibers to the geometry of strong-stable
manifolds. Specifically, for suitable $\rho,r>0$ and a set
$E\subset A$ with $\mu(E)>0$, we have
\[
\Cs(x)\subseteq \Wss_{\phi,r}(x)\cap A
\qquad(x\in E).
\]
Here $\Wss_{\phi,r}(x)$ is a local strong-stable curve of
intrinsic radius $r$ through $x$ in the surrounding state space.
For Lorenz dynamics, its intersection with $A$ is totally
disconnected, so the contained fiber $\Cs(x)$ is also totally
disconnected for $x\in E$. Cantor sets and singletons are
examples of totally disconnected sets. This inclusion and the 
total disconnectedness are justified in Sec.~\ref{sec:lorenz-application}.

We write $h_\mu(f)$ for the Kolmogorov-Sinai metric entropy.
For an invertible measure-preserving map,
$h_\mu(f^{-1})=h_\mu(f)$. The theorem concerns cases in which this
common entropy value is positive.

We ask whether a generator admits a compact attracting set $\Lambda$
on which it realizes the prescribed reversed motion. We express
attraction of the generator system through a trapping neighborhood $U$. 
Specifically, for the generator diffeomorphism $g:M\to M$, let $U\subset M$
be a nonempty open neighborhood of $\Lambda$, and write
$\overline U$ for its closure in $M$. We assume that $\overline U$ is compact
and require
\begin{equation}
g(\overline U)\subset U,\qquad
\Lambda=\bigcap_{n\ge0}g^n(\overline U).
\label{eq:attractor}
\end{equation}
The first condition makes $U$ a trapping neighborhood. The successive
images of $\overline U$ form nested nonempty compact sets, so $\Lambda$
is nonempty, compact, and invariant, $g(\Lambda)=\Lambda$.
Trajectories starting in $U$ approach $\Lambda$.

Attraction concerns distance to $\Lambda$ in the full generator
state space, including all latent variables, while trajectories within
$\Lambda$ may still be chaotic.
This trapping condition is important for the proof: every complete
unstable leaf through a regular generator state in $\Lambda$ stays in
$\Lambda$. The readout $\pi$ can therefore be applied along each whole
leaf. Lemma~\ref{lem:trapped} proves this containment from
Eq.~\eqref{eq:attractor}.

With these definitions, we first state the obstruction for invertible
maps and then derive a corollary for autonomous flows.

\begin{theorem}[Positive-measure stable-fiber obstruction]
\label{thm:main}
Let $f:A\to A$ be a homeomorphism of a compact metric space $A$
preserving an ergodic Borel probability measure $\mu$ with
$h_\mu(f)>0$. Suppose there are $\rho>0$ and a Borel set
$E\subset A$ with $\mu(E)>0$ such that $\Cs(x)$ is totally
disconnected for every $x\in E$.

Let $g:M\to M$ be any $C^2$ diffeomorphism on a finite-dimensional
compact smooth manifold $M$ without boundary. If
$\Lambda\subset M$ is an attracting set satisfying
Eq.~\eqref{eq:attractor}, then no continuous surjection
$\pi:\Lambda\to A$ can satisfy
\begin{equation}
\pi\circ g=f^{-1}\circ\pi
\qquad\text{on }\Lambda.
\label{eq:factor}
\end{equation}
\end{theorem}

Equation~\eqref{eq:factor} is the \emph{factor relation} for time reversal:
one generator step produces one backward step in the target dynamics at
every point of $\Lambda$. Invertibility and invariance give
\[
\pi(g^k(z))=f^{-k}(\pi(z))\qquad(z\in\Lambda,\ k\in\mathbb Z).
\]
The identity holds at every point of $\Lambda$.

Ergodicity is assumed only for the target measure $\mu$.
Under the factor relation, the proof constructs an ergodic invariant
measure $\nu$ on $\Lambda$ with $\pi_*\nu=\mu$.
Neither uniqueness of this $\nu$ nor full support on $\Lambda$ is required.

Total disconnectedness is required of the entire intrinsic stable
fiber $\Cs(x)\subset A$ for each $x\in E$, where $\mu(E)>0$;
$E$ need not be invariant. The hypotheses allow zero Lyapunov exponents
in some directions and hidden chaotic dynamics in the generator.
No differentiability or injectivity is required of the continuous readout.

Theorem~\ref{thm:main} assumes that the generator's ambient manifold
$M$ is compact and without boundary, matching the hypotheses of Ledrappier-Young's
Theorem~B, used in the proof~\cite{LY1985}.
Many autonomous flows, however, are defined on Euclidean space
$\Real^D$, which has no boundary but is noncompact. We can still obtain
an obstruction in this setting: the construction below preserves the
dynamics near $\Lambda$ and the trapping of nearby states while
modifying the vector field farther away. It produces a flow on a new
compact manifold without boundary, of the same dimension. The following
corollary therefore allows the original state space to be noncompact,
while retaining a compact attracting set $\Lambda$ and the stated
trapping assumptions.

\begin{corollary}[Autonomous flows]
\label{cor:flow}
Suppose a target flow $\phi_t|_A$ satisfies the target hypotheses of
Theorem~\ref{thm:main} for $f=\phi_\tau|_A$ at some $\tau>0$.
Let $G$ be a $C^2$ autonomous vector field on a finite-dimensional
smooth manifold $M$ without boundary, which need not be compact,
and let $\Psi_t$ denote its flow. Suppose $\Lambda\subset M$ is
compact and flow invariant and satisfies Eq.~\eqref{eq:attractor}
for $g=\Psi_\tau$ and an open neighborhood $U$ with compact closure.
Assume $\Psi_t$ is defined on $\overline U$ for $0\le t\le\tau$.
Then no continuous surjection $\pi:\Lambda\to A$ can satisfy
\begin{equation}
\pi\circ\Psi_t=\phi_{-t}\circ\pi
\qquad\text{on }\Lambda,\ t\in\Real.
\label{eq:flowfactor}
\end{equation}
\end{corollary}

\begin{proof}
Suppose such a realization exists. We will construct a flow on a
compact manifold without boundary that preserves the same attracting
realization. Its time-$\tau$ map will then contradict
Theorem~\ref{thm:main}. If $M$ is already compact, no modification is
needed: its flow is complete, and $g=\Psi_\tau$ is a $C^2$
diffeomorphism satisfying Eq.~\eqref{eq:factor}.

For noncompact $M$, we must preserve both the motion on $\Lambda$
and the trapping of nearby states. The compact orbit tube
\[
K=\{\Psi_t(z):z\in\overline U,\ 0\le t\le\tau\}
\]
contains every state reached from $\overline U$ during one sampling
interval. Choose a compact smooth domain $N\subset M$ with smooth
boundary $\partial N$ and $K\subset\operatorname{int}N$, where
$\operatorname{int}N=N\setminus\partial N$ denotes its interior in $M$.
Such a domain exists because $K$ is compact.

Choose a modified $C^2$ vector field $G'$ on $N$ that equals $G$
on a neighborhood of $K$ and vanishes on a neighborhood of $\partial N$.
This is achieved by multiplying $G$ by a smooth cutoff equal to one
near $K$ and zero near the boundary.

Theorem~\ref{thm:main} requires an ambient manifold without boundary.
Glue two copies of $N$ along their boundaries to obtain such a compact
smooth manifold $M'$, of the same dimension as $M$. Extend $G'$ by zero
on the second copy. Since the field vanishes near the shared boundary,
it is $C^2$ on $M'$ and has a complete flow $\Psi'_t$. Consequently
$g'=\Psi'_\tau$ is a $C^2$ diffeomorphism of $M'$.

Identifying $\overline U$ with its copy in $M'$, uniqueness of
solutions gives $g'=\Psi_\tau$ on $\overline U$.
Because $\Psi_\tau(\overline U)\subset U$, all subsequent sampled
images also agree. Thus $g'$ has the same trapping neighborhood and
attracting set $\Lambda$. Every trajectory in $\Lambda$ stays in
$\Lambda\subset K$, so $\Psi'_t=\Psi_t$ there for all $t\in\Real$.
In particular, the same readout satisfies
\[
\pi\circ g'=f^{-1}\circ\pi\qquad\text{on }\Lambda.
\]
This attracting realization on a compact boundaryless manifold
contradicts Theorem~\ref{thm:main}.
\end{proof}

\section{Proof of Theorem~\ref{thm:main}}
\label{sec:proof}

We argue by contradiction. Suppose such a realization exists. The argument has four
steps: lift the measure, collapse complete leaves whose outputs lie in $E$,
extend the fiber condition by target ergodicity, and apply the Pinsker theorem.

\subsection{An ergodic lift with positive output entropy}

We construct an ergodic invariant generator measure $\nu$ projecting
to $\mu$. With respect to this measure, the prescribed output has
positive metric entropy, and therefore so does the generator.

We first choose a probability measure $\nu_0$ on $\Lambda$ with
$\pi_*\nu_0=\mu$, using a measurable right inverse of $\pi$.
This initial measure need not be invariant or ergodic.
By compactness, its time averages have a weakly convergent
subsequence:
\[
\bar\nu
=\lim_{j\to\infty}\frac{1}{N_j}
  \sum_{k=0}^{N_j-1}g^k_*\nu_0,
\qquad N_j\to\infty.
\]
Time averaging makes $\bar\nu$ invariant under $g$, while the
factor relation and invariance of $\mu$ ensure that
$\pi_*\bar\nu=\mu$.

The ergodic decomposition theorem expresses $\bar\nu$ as a
weighted mixture of ergodic invariant probability measures:
\[
\bar\nu=\int\nu_\omega\,d\kappa(\omega),
\]
where $\kappa$ gives the weights of the components.
Projecting this mixture gives
\[
\mu=\pi_*\bar\nu
   =\int\pi_*\nu_\omega\,d\kappa(\omega).
\]
Each projected measure $\pi_*\nu_\omega$ is invariant under
$f^{-1}$ by the factor relation. Since $\mu$ is ergodic,
it cannot be a nontrivial mixture of distinct $f^{-1}$-invariant
probability measures. Therefore $\pi_*\nu_\omega=\mu$ for
$\kappa$-almost every $\omega$. Choose one such component
and denote it by $\nu$; it is ergodic, invariant under $g$,
and has the prescribed output distribution $\mu$.

For every finite measurable partition $\alpha$ of $A$, the
factor relation and $\pi_*\nu=\mu$ give
\begin{equation}
h_\nu(g,\pi^{-1}(\alpha))=h_\mu(f^{-1},\alpha).
\label{eq:partitionentropy}
\end{equation}
Thus the output factor has entropy
\begin{equation}
h_{\rm out}=h_\mu(f^{-1})=h_\mu(f)>0,
\qquad h_\nu(g)\ge h_{\rm out},
\label{eq:houtpos}
\end{equation}
using invariance of metric entropy under inversion~\cite{Walters1982}.

\subsection{Whole-leaf collapse at centers over \texorpdfstring{$E$}{E}}
\label{sec:localcollapse}

We prove that $\pi$ is constant on every complete unstable leaf through
a regular generator state whose output belongs to $E$.

Oseledets' and Pesin's theorems give a $g$-invariant Borel set
$\Gamma\subset\Lambda$ of full $\nu$-measure on which the Lyapunov
splittings and unstable manifolds are defined
\cite{Oseledets1968,Pesin1976,BarreiraPesin2007}. We call these states regular. 
Since $\nu$ is ergodic, Ruelle's inequality (Ref.~\cite{Ruelle1978}, Theorem~2(b), pp.~84--85)
and Eq.~\eqref{eq:houtpos} imply that at least one Lyapunov exponent is positive at $\nu$-almost every
generator state. Zero exponents in other directions are allowed. 
The complete unstable leaf through $z\in\Gamma$ is
\begin{multline}
\Wu_g(z)=\bigl\{z'\in M:\\
\limsup_{n\to\infty}\tfrac1n\log\!\bigl(d_M(g^{-n}(z'),g^{-n}(z))\bigr)<0\bigr\},
\label{eq:completeleaf}
\end{multline}
where we use the convention $\log 0=-\infty$ to include $z$ itself in the leaf.
Here $d_M$ is the ambient Riemannian distance on the generator manifold
$M$, introduced in Sec.~\ref{sec:setup}. The leaf consists of generator
states whose backward trajectories approach that of $z$ exponentially
(Ref.~\cite{Ruelle1979}, Remark~6.4(b), p.~53, applied to $g^{-1}$).

The readout $\pi:\Lambda\to A$ is defined on $\Lambda$, whereas
$\Wu_g(z)$ is initially defined as a subset of $M$. To apply $\pi$
along the whole leaf, we first show that every point of $\Wu_g(z)$
lies in $\Lambda$.

\begin{lemma}[Trapping of complete unstable leaves]
\label{lem:trapped}
For every $z\in\Gamma$, $\Wu_g(z)\subset\Lambda$.
\end{lemma}
\begin{proof}
Fix $z\in\Gamma$ and an arbitrary $z'\in\Wu_g(z)$.
Since $\Lambda$ is invariant, $g^{-n}(z)\in\Lambda$, and
\[
d_M(g^{-n}(z'),\Lambda)
\le d_M(g^{-n}(z'),g^{-n}(z))\longrightarrow0.
\]
Compactness of $\Lambda\subset U$ implies that points sufficiently
close to $\Lambda$ lie in $U$. Hence there is an integer $N$ such that
$g^{-n}(z')\in U$ for every $n\ge N$. Applying $g^n$ gives
\[
z'\in g^n(U)\subset g^n(\overline U)\qquad(n\ge N).
\]
By Eq.~\eqref{eq:attractor},
$g(\overline U)\subset U\subset\overline U$, so the sets
$g^n(\overline U)$ are nested and decreasing. For any $k\ge0$,
choose $n\ge\max\{N,k\}$. Then
\[
z'\in g^n(\overline U)\subset g^k(\overline U).
\]
Thus $z'$ belongs to every such set, and Eq.~\eqref{eq:attractor}
identifies their full intersection:
\[
z'\in\bigcap_{k\ge0}g^k(\overline U)=\Lambda.
\]
Since $z'$ was arbitrary, $\Wu_g(z)\subset\Lambda$.
\end{proof}

\emph{The leaf and its local neighborhoods.}
The unstable-manifold theorem gives $\Wu_g(z)$ a connected manifold
structure. Its inclusion into $M$ is an \emph{injective immersion}:
it is one-to-one with injective derivative. The whole leaf need not
be embedded; portions far apart along the leaf may approach one another
in $M$. We therefore use its \emph{intrinsic topology}, defined by
local coordinate disks within the leaf. These small disks, also called
plaques, are embedded in $M$. Openness below means openness within
the leaf $\Wu_g(z)$, not within $M$. The inclusion into $M$ is continuous,
so the readout restricted to the leaf is continuous in this topology.

\emph{Backward closeness on local patches.}
Fix a regular generator state $z\in\Gamma$ with
$x=\pi(z)\in E$. To place the image of a generator patch in
$C^s_\rho(x)$, we need both convergence and closeness at every
forward target iterate. The factor relation converts these
requirements into backward-time conditions in the generator.

For every $q\in\Wu_g(z)$ and every $\eta>0$, there is a connected
open leaf neighborhood $V_q$ of $q$ such that, for every $v\in V_q$,
\begin{equation}
\begin{aligned}
d_M(g^{-m}(v),g^{-m}(q))&<\eta\qquad(m\ge0),\\
d_M(g^{-m}(v),g^{-m}(q))&\longrightarrow0
\qquad(m\to\infty).
\end{aligned}
\label{eq:localbackward}
\end{equation}
Indeed, $q$ lies in the intrinsic interior of a finite forward
image of a local unstable disk based on the regular backward
orbit of $z$. Choose an open leaf patch $P$ around $q$ with
compact leafwise closure inside that image.
The unstable-manifold estimates give backward convergence to
the orbit of $q$ uniformly over $P$
(Ref.~\cite{Ruelle1979}, Theorems~6.1(b) and 6.3,
pp.~50--53, applied to $g^{-1}$).
Thus the distance in Eq.~\eqref{eq:localbackward} is below $\eta$
for every $v\in P$ and every $m\ge m_0$, for some finite $m_0$.

For the finitely many earlier iterates $0\le m<m_0$,
continuity of $g^{-m}$ lets us shrink $P$ to a connected open
leaf neighborhood $V_q$ on which the same bound holds.
The later-time bounds remain valid because $V_q\subset P$.
This proves Eq.~\eqref{eq:localbackward}.
The point $q$ need not itself be regular; the construction uses
the regular backward orbit of $z$.

\emph{Collapse of a patch around $z$.}
Uniform continuity of $\pi:\Lambda\to A$ supplies $\eta>0$
such that $d_M(a,b)<\eta$ implies
$d_A(\pi(a),\pi(b))<\rho$ for $a,b\in\Lambda$.
Take the patch $V_z$ given by Eq.~\eqref{eq:localbackward}
for this $\eta$.
For every $v\in V_z$, the factor relation gives
\begin{align}
d_A(f^m(\pi(v)),f^m(x))
&=d_A(\pi(g^{-m}(v)),\pi(g^{-m}(z)))<\rho,
\nonumber\\
d_A(f^m(\pi(v)),f^m(x))&\longrightarrow0.
\label{eq:intofiber}
\end{align}
The bound holds for every $m\ge0$, and the limit is as
$m\to\infty$. Thus $\pi(V_z)\subset C^s_\rho(x)$.
Because $V_z$ is connected and $\pi$ is continuous,
$\pi(V_z)$ is connected and contains $x$.
The target fiber $C^s_\rho(x)$ is totally disconnected, so
\[
\pi(V_z)=\{x\}.
\]

\emph{Extending the patch collapse to the whole leaf collapse.}
Define the set of generator states on this leaf that output $x$:
\begin{equation}
Z_x=\{q\in\Wu_g(z):\pi(q)=x\}.
\label{eq:levelset}
\end{equation}
For any $q\in Z_x$, choose its connected open leaf neighborhood
$V_q$ from Eq.~\eqref{eq:localbackward}, using the same $\eta$.
Since $\pi(q)=x$, the calculation in Eq.~\eqref{eq:intofiber},
with $q$ replacing $z$, gives $\pi(V_q)\subset C^s_\rho(x)$.
Connectedness again implies $\pi(V_q)=\{x\}$, hence
$V_q\subset Z_x$.

Every point of $Z_x$ therefore has an open leaf neighborhood
contained in $Z_x$, so $Z_x$ is open in $\Wu_g(z)$. It is also 
closed in the leaf: continuity of $\pi$ implies that the set of leaf 
points mapped to the closed singleton $\{x\}$ is closed.
Since $z\in Z_x$ and the leaf is connected,
$Z_x=\Wu_g(z)$. Consequently,
\begin{equation}
\pi(\Wu_g(z))=\{\pi(z)\}
\qquad(z\in\Gamma,\ \pi(z)\in E).
\label{eq:goodleafcollapse}
\end{equation}

\subsection{Positive measure suffices}
\label{sec:promotion}
Equation~\eqref{eq:goodleafcollapse} establishes whole-leaf
collapse for regular generator states whose outputs lie in $E$.
We now extend this conclusion to $\nu$-almost every generator
state. To do so, we first extend the target fiber condition
from $E$ to $\mu$-almost every target point, then apply the
whole-leaf argument of Sec.~\ref{sec:localcollapse}.

For every $x\in A$ and integer $n\ge0$, the definition in
Eq.~\eqref{eq:fiber} gives
\begin{equation}
f^n\bigl(\Cs(x)\bigr)\subseteq\Cs(f^n(x)).
\label{eq:fibertransport}
\end{equation}
Indeed, if $y\in\Cs(x)$, then for every $k\ge0$,
\[
d_A\bigl(f^k(f^n(y)),f^k(f^n(x))\bigr)
=d_A(f^{n+k}(y),f^{n+k}(x))<\rho,
\]
and these distances tend to zero as $k\to\infty$.
If $f^n(x)\in E$, the right-hand side of
Eq.~\eqref{eq:fibertransport} is totally disconnected, so its subset
$f^n(\Cs(x))$ is also totally disconnected. Since $f^n$ restricts
to a homeomorphism from $\Cs(x)$ onto its image, $\Cs(x)$ must be
totally disconnected as well. Thus the property propagates
\emph{backward} from the visited point $f^n(x)$ to $x$.

The given set $E$ is Borel and has $\mu(E)>0$. By ergodicity and
Birkhoff's theorem, almost every target orbit visits $E$ with
frequency $\mu(E)>0$. Consequently, the Borel set
\begin{equation}
X_E=\bigcup_{n\ge0}f^{-n}(E),\qquad\mu(X_E)=1,
\label{eq:targetvisits}
\end{equation}
consists of points with totally disconnected fibers: for each
$x\in X_E$, some finite iterate $f^n(x)$ belongs to $E$, so the preceding
argument applies. The fiber condition therefore holds on a Borel set
of full $\mu$-measure.

Return now to the generator. Since $\pi_*\nu=\mu$ and
$\nu(\Gamma)=1$, the Borel set
\begin{equation}
\Omega=\Gamma\cap\pi^{-1}(X_E),\qquad\nu(\Omega)=1,
\label{eq:goodreturns}
\end{equation}
consists of regular states $z$ for which $\Cs(\pi(z))$ is totally
disconnected. Applying the argument of
Sec.~\ref{sec:localcollapse} at each such state gives
\begin{equation}
\pi(\Wu_g(z))=\{\pi(z)\}\qquad(z\in\Omega).
\label{eq:leafcollapse}
\end{equation}

\subsection{The output belongs to the zero-entropy factor}

Equation~\eqref{eq:leafcollapse} shows that the readout $\pi$ cannot
distinguish generator states $q$ and $q'$ lying on the same complete unstable
leaf $\Wu_g(z)$ for any $z\in\Omega$. Ledrappier-Young's theorem connects this
loss of distinguishability to zero entropy: measurable information
that is constant along unstable leaves belongs to the
zero-entropy factor. We apply this result to show that the
output has zero metric entropy, even though the generator
itself has positive metric entropy.

Let $\Bu(g)$ be the completed $\sigma$-algebra of measurable unions of complete
unstable leaves, and let $\Pin(g)$ be the Pinsker $\sigma$-algebra, the maximal
zero-entropy factor. Ledrappier-Young's Theorem~B (Ref.~\cite{LY1985}, p.~513) states
\begin{equation}
\Bu(g)=\Pin(g)\qquad\text{modulo }\nu.
\label{eq:LY}
\end{equation}
Its hypotheses are a $C^2$ diffeomorphism of a compact Riemannian manifold
without boundary and an invariant Borel probability measure. It requires neither absolute
continuity of $\nu$ nor nonzero exponents in every direction.

We apply this result to the smooth ambient system $(M,g,\nu)$,
regarding $\nu$ as a measure on $M$ with $\nu(\Lambda)=1$.
The attracting set $\Lambda$ itself need not be a smooth manifold.

Let $\alpha=\{B_1,\ldots,B_r\}$ be any finite Borel partition
of the output space $A$. Its preimage under the readout is
\[
\beta=\pi^{-1}(\alpha)
     =\{\pi^{-1}(B_1),\ldots,\pi^{-1}(B_r)\}.
\]
Here $\pi^{-1}(B_i)=\{z\in\Lambda:\pi(z)\in B_i\}$ denotes
a preimage; no invertibility of $\pi$ is required.
The partition $\beta$ groups generator states according to
their output categories. Since $\nu(\Lambda)=1$, it is also
a partition of $M$ modulo a null set.

We show that every atom of $\beta$ belongs to $\Bu(g)$
modulo $\nu$. Let
\[
\Omega^u=\bigcup_{z\in\Omega}\Wu_g(z).
\]
By Lemma~\ref{lem:trapped},
$\Omega\subset\Omega^u\subset\Lambda$.
Since $\nu(\Omega)=1$, the complement of $\Omega^u$ lies in
the Borel null set $M\setminus\Omega$. Thus $\Omega^u$ is
measurable in the completion of $\nu$.
Equation~\eqref{eq:leafcollapse} says that $\pi$ is constant
on every leaf in this union, and intersecting complete leaves
coincide. Consequently,
$\pi^{-1}(B_i)\cap\Omega^u$ is a measurable union of whole
unstable leaves. It differs from $\pi^{-1}(B_i)$ only within
$M\setminus\Omega$, a set of measure zero. Hence
\begin{equation}
\pi^{-1}(B_i)\in\Bu(g)=\Pin(g)
\quad\text{modulo }\nu,\qquad i=1,\ldots,r.
\label{eq:inclusion}
\end{equation}
By Eq.~\eqref{eq:LY}, every atom of $\beta$ therefore belongs
to the Pinsker $\sigma$-algebra. Recall that the Shannon entropy
of the partition $\beta$ is
\[
H_\nu(\beta)
=-\sum_{D\in\beta}\nu(D)\log\nu(D),
\qquad 0\log 0=0.
\]
The associated entropy rate under $g$ satisfies
\[
h_\nu(g,\beta)
=\lim_{n\to\infty}\frac1n
 H_\nu\!\left(\bigvee_{k=0}^{n-1}g^{-k}\beta\right)
=0,
\]
where $\bigvee_{k=0}^{n-1}g^{-k}\beta$ is the common
refinement that records the partition categories over
$n$ successive iterates.

The factor relation and $\pi_*\nu=\mu$ give, as established
in Eq.~\eqref{eq:partitionentropy},
\begin{equation}
h_\mu(f^{-1},\alpha)
=h_\nu(g,\pi^{-1}(\alpha))
=0.
\label{eq:entropy-factor}
\end{equation}
Since $\alpha$ was an arbitrary finite Borel partition
of $A$, the definition of metric entropy gives
\[
h_{\rm out}=h_\mu(f^{-1})
=\sup_\alpha h_\mu(f^{-1},\alpha)=0.
\]
This contradicts Eq.~\eqref{eq:houtpos} and proves the
theorem.\hfill$\square$

\section{Lorenz dynamics and reservoir-computing examples}
\label{sec:lorenz}

\begin{figure*}[t]
\centering
\includegraphics[width=\textwidth]{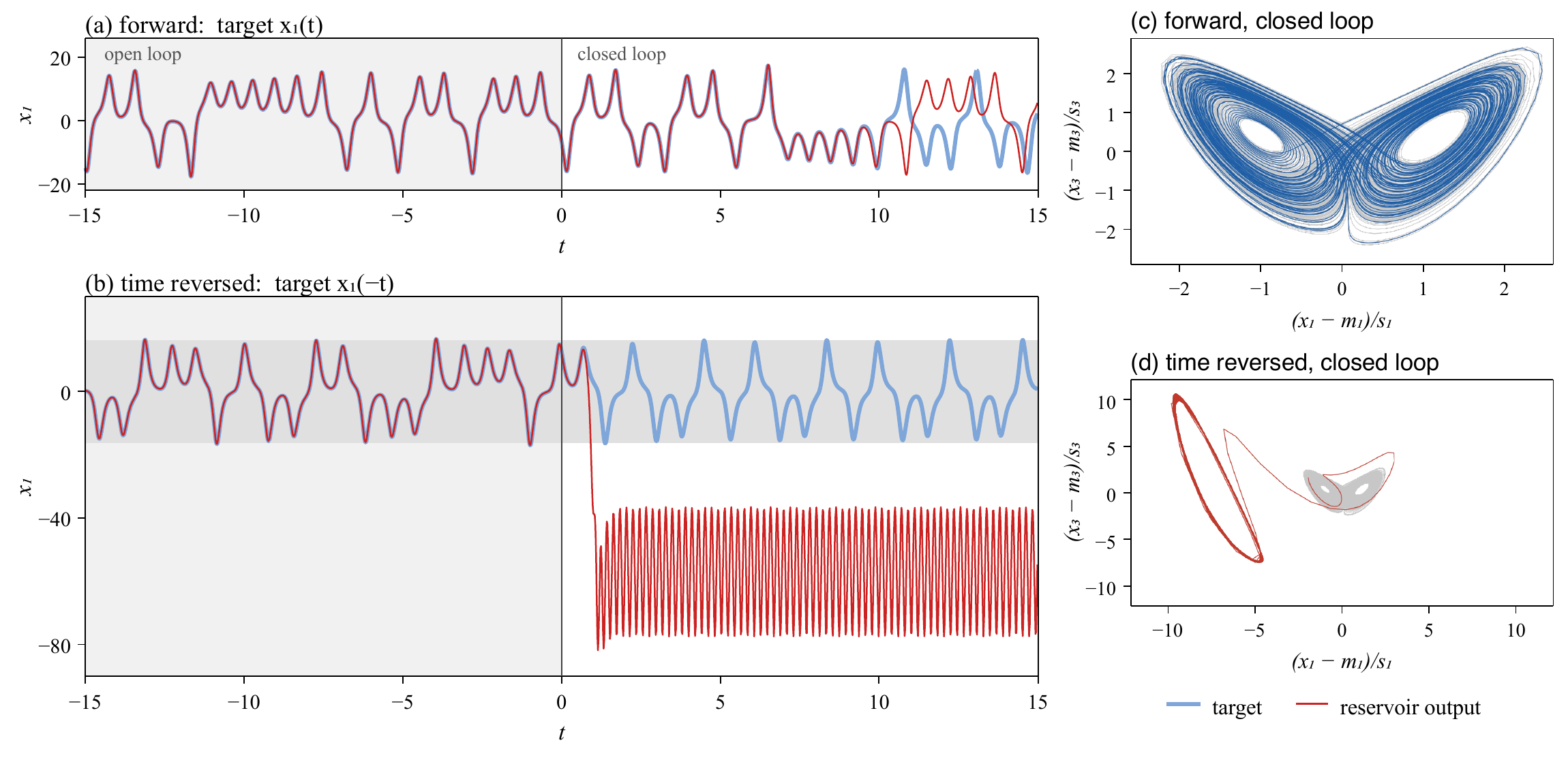}
\caption{Forward and reversed Lorenz signals from paired reservoir
networks ($N=500$, seed $0$). (a,b)~First-coordinate signals in physical
Lorenz units: target in blue, reservoir output in red; the displayed
legend applies to these panels. Light shading marks teacher forcing;
autonomous predictions begin at $t=0$. Vertical ranges differ; the darker
band in (b) marks the target range. (c,d)~Autonomous outputs over a nominal
$80$-time-unit interval: forward in blue, reversed in red, with the target
reference trajectory in gray. The axes show centered and standardized
coordinates $(x_j-m_j)/s_j$. The two panel pairs use separate records and
readout fits; Appendix~\ref{app:numerics} defines $m_j,s_j$ and gives
their protocols.}
\label{fig:lorenz}
\end{figure*}

\subsection{Application to the classical Lorenz attractor}
\label{sec:lorenz-application}

We verify the hypotheses of Theorem~\ref{thm:main} for the classical
Lorenz system,
\begin{equation}
\begin{aligned}
\dot x_1&=10(x_2-x_1),&\quad \dot x_2&=x_1(28-x_3)-x_2,\\
\dot x_3&=x_1x_2-\tfrac83x_3.
\end{aligned}
\label{eq:lorenz}
\end{equation}
Its singular-hyperbolic attractor $A$ carries an SRB measure $\mu$
with full support and positive metric entropy~\cite{Tucker2002,APPV2009}.
The flow is mixing with respect to $\mu$~\cite{LMP2005}, so every
sampled map $f=\phi_\tau|_A$, $\tau>0$, is ergodic. The entropy
scaling formula~\cite{Abramov1959} gives
\begin{equation}
h_\mu(f)=\tau h_\mu(\phi_1)>0.
\end{equation}
We use the Euclidean metric on $A$.

The remaining hypothesis follows from the attractor's stable geometry.
Let $\Wss_{\phi,r}(x)$ denote a local strong-stable curve through $x$,
with radius $r$ measured along the curve. The Lorenz attractor contains
no nondegenerate strong-stable interval~\cite{Morales2007}. Since any
connected subset of this one-dimensional curve containing two distinct
points contains an interval, $\Wss_{\phi,r}(x)\cap A$ is totally
disconnected for every $x\in A$.

To connect this geometry to the sampled dynamics, fix $\tau>0$.
There exist $\rho,r>0$ and a Borel set $E\subset A$ with $\mu(E)>0$
such that
\begin{equation}
\Cs(x)\subset\Wss_{\phi,r}(x)\cap A
\qquad(x\in E).
\label{eq:fiberid}
\end{equation}
Choose $E$ as the $f$-recurrent points in a small open flow box
meeting $A$ whose orbits cross the central part of an adapted section
(Ref.~\cite{APPV2009}, Lemma~2.10); full support and recurrence give
$\mu(E)>0$. Uniform finite-time estimates extend sampled closeness
and convergence to continuous time. Theorem~3.1, applied after a
section crossing with identity time reparametrization, and flow
invariance then place each $x'\in\Cs(x)$,
$x\in E$, on the strong-stable manifold of $\phi_\theta(x)$ for some
$\theta$. Synchronous convergence and recurrence imply
$\phi_\theta(x)=x$. Proposition~2.1 supplies uniform local plaques:
strong-stable contraction forces eventual entry, while sufficiently
close forward trajectories cannot cross a plaque boundary, which
stays uniformly separated from its base point. Choosing $\rho$
sufficiently small therefore proves Eq.~\eqref{eq:fiberid}.
Thus $\Cs(x)$ is a subset of a totally disconnected set for every
$x\in E$. Together with positive entropy and ergodicity, this verifies
the hypotheses of Theorem~\ref{thm:main}. The sampled reversed Lorenz
dynamics and the reversed flow therefore admit no exact attracting
realizations of the kinds excluded by Theorem~\ref{thm:main} and
Corollary~\ref{cor:flow}.

\subsection{Matched reservoir-computing examples}
\label{sec:rc}

To illustrate the contrast between forward and reversed autonomous
generation, we compare reservoir-computing networks trained on the same
Lorenz trajectory in opposite time directions. Within each pair, the
reservoir seed and hyperparameters are matched, with a separately fitted
readout layer for each time direction. Figure~\ref{fig:lorenz} compares teacher-forced prediction
with autonomous rollout; Appendix~\ref{app:numerics} records the
simulation protocol.

Both networks' outputs closely track their prescribed signals during
teacher forcing. However, after the feedback loop is closed, the
forward-time example autonomously generates output that retains the
target's two-lobed geometry over the displayed interval, whereas the
reversed example develops a large oscillation. These finite-time
observations do not establish asymptotic attraction.

The theorem rules out exact compact attracting realizations of
time-reversed Lorenz dynamics within the stated smooth generator class.
The numerical examples concern finite training data and finite-time
generation. The layered geometry discussed in the Introduction suggests
that stabilization preserving the prescribed motion may require fine
spatial variation between nearby parts of the target set. This offers
intuition for the practical difficulty of autonomous generation of
time-reversed data.

\section{Discussion}
\label{sec:discussion}

Our theorem identifies a limit on learning dynamics for stable autonomous
generation. For dynamics satisfying the stated geometric and entropy
hypotheses, no model in the specified finite-dimensional smooth generator
class can realize the exact reversed dynamics on a compact attracting set
through a continuous readout. Increasing the finite latent dimension or
changing the training procedure cannot remove this obstruction. Accurate
fitting and finite-time prediction remain distinct from this requirement
of exact attracting realization.

The same limitation applies to autonomous control through additional
latent variables. The controlled system and its controller together form
a generator. If their joint state must approach a compact attracting set
while the output preserves the prescribed reversed motion exactly, the
same obstruction applies under the stated assumptions. Learning and
motion-preserving control therefore face a common geometric constraint
on the dynamical realizations they can produce.

This result provides a theoretical probe of why the arrow of time can
matter for autonomous generation. Periodic motion and quasiperiodic
rotations admit attracting realizations in both directions, whereas
the Lorenz application exhibits an obstruction in the reversed direction.
This distinction motivates broader questions about learned behavior,
including why some musical patterns are readily produced backward while
others are difficult. The theorem does not classify the difficulty of
finite musical sequences; connecting such behavior to the present result
requires further work on approximation, memory, and representation.
It motivates distinguishing limitations of data and training from
constraints on the existence of a stable autonomous realization.

\begin{acknowledgments}
I wish to thank the Allen Institute founder, Paul G. Allen, 
for his vision, encouragement, and support. I am especially 
grateful to my Ph.D. mentor, Edward Ott, for teaching 
me nonlinear dynamics and shaping the way
I think about dynamical systems. His guidance continues
to influence the questions I ask and how I approach them,
and much of what I learned from him underlies this work.
\end{acknowledgments}

\appendix
\section{Numerical methods}
\label{app:numerics}

The Lorenz equations were integrated by fourth-order Runge--Kutta
with step $\Delta t=0.02$, starting from $(1,1,1)$. After a
$5000$-sample transient (approximately $100$ time units), each
retained record was standardized coordinatewise as
$\widetilde x_j=(x_j-m_j)/s_j$, where $m_j$ and $s_j$ are the mean
and standard deviation of the entire retained record, including
the portion used after training. The reversed sequence is this
same standardized record in reverse order.

Each reservoir has $N=500$ nodes. Its matrix $B$ is a random sparse
matrix with density $3/N$, nonzero weights drawn uniformly from
$[-1,1]$, and spectral radius rescaled to $0.9$. Entries of
$W_{\rm in}\in\mathbb R^{N\times3}$ are uniform on $[-0.1,0.1]$.
The forward and reversed networks use the same reservoir seed
and hyperparameters, with separately fitted readouts. For a
standardized target input $u_k\in\mathbb R^3$, the driven update is
\begin{equation}
z_{k+1}=\tanh\!\left(Bz_k+W_{\rm in}u_k\right),\qquad
\widehat u_{k+1}=W_{\rm out}q(z_{k+1}),
\label{eq:driven-reservoir}
\end{equation}
where $q(z)=[z;z^2]\in\mathbb R^{2N}$ and both squaring and
$\tanh$ act componentwise. This is a nonleaky discrete-time
update; no additional step coefficient or bias is used.

In each direction, training uses the first $25000$ samples of
the ordered record. These training windows therefore come from
opposite ends of the full record. Starting from $z_0=0$, the first
$24999$ inputs produce states $z_1,\ldots,z_{24999}$. After discarding
the first $1000$ states, the $23999$ training pairs are
$(q(z_{k+1}),u_{k+1})$, $k=1000,\ldots,24998$.
Writing their feature and target columns as $X$ and $Y$, respectively,
the readout minimizes
\begin{equation}
\|Y-WX\|_F^2+\eta\|W\|_F^2,
\qquad \eta=10^{-7},
\label{eq:ridge-reservoir}
\end{equation}
without dividing the data-fit term by the number of samples.
Once the output is fed back, the autonomous update is
\begin{equation}
z_{k+1}=\tanh\!\left(Bz_k+W_{\rm in}W_{\rm out}q(z_k)\right).
\label{eq:closedloop}
\end{equation}

The two panel pairs in Fig.~\ref{fig:lorenz} were generated
separately, both with seed $0$. Panels (a,b) use a $31000$-sample
record. After training, each reservoir is reset to zero and driven
by $1300$ subsequent inputs. The $750$ teacher-forced predictions
immediately preceding closure and the next $750$ autonomous predictions are shown,
giving nominal $15$-time-unit windows on either side of loop
closure. The first autonomous prediction is read from the final
driven state; subsequent predictions use output feedback. The
plotted first coordinate is converted back to physical Lorenz units.

Panels (c,d) use a $30510$-sample record with its own normalization
and readout fit. Following training, the reservoir is reset to zero
and driven by $500$ inputs; one further true input is consumed
to form the first autonomous prediction. The rollout contains
$4000$ predictions, spanning nominally $80$ time units. The phase
portraits show standardized first and third coordinates, together
with the first $20000$ standardized target points as a reference.
The implementation stops a rollout if an output is nonfinite or
has a standardized coordinate of magnitude greater than $10^3$,
clips that output at this threshold, and marks subsequent values
as missing.
\bibliography{references}
\end{document}